\documentclass[11pt]{article}  

\usepackage{amsfonts,amsmath,amssymb,amsthm}
\usepackage[margin=1in]{geometry}
\usepackage{verbatim}
\usepackage{times}
\usepackage{bm}
\usepackage{xspace}
\usepackage{graphicx}
\usepackage{hyperref}

\newtheorem{theorem}{Theorem}[section]
\newtheorem{lemma}[theorem]{Lemma}

\newtheorem{remark}{Remark}

\theoremstyle{definition}
\newtheorem{definition}{Definition}

\renewcommand{\Pr}{\mathop{\bf Pr\/}}

\newcommand{\E}{\mathop{\bf E\/}}

\newcommand{\R}{\mathbb R}

\newcommand{\opt}{\mathsf{opt}}

\newcommand{\eps}{\epsilon}

\newcommand{\calF}{{\cal F}}

\newcommand{\calR}{{\cal R}}

\newcommand{\bb}{\boldsymbol{b}}
\newcommand{\bp}{\boldsymbol{p}}
\newcommand{\bq}{\boldsymbol{q}}

\newcommand{\bv}{\boldsymbol{v}}

\newcommand{\bS}{\boldsymbol{S}}

\newcommand{\abs}[1]{\left\lvert #1 \right\rvert}

\begin{document}

\title{\bf Simple and Nearly Optimal Multi-Item Auctions}
\author{Yang Cai\thanks{Supported by NSF Award CCF-0953960 (CAREER) and CCF-1101491. Part of this work was done while the author was visiting Microsoft Research, Redmond.}\\ EECS, MIT\\
\tt{ycai@csail.mit.edu} \\
\and 
Zhiyi Huang\thanks{ This work was supported in part by an ONR MURI Grant N000140710907. Part of this work was done while the author was visiting Microsoft Research, Redmond.}\\ 
University of Pennsylvania\\
\tt{ hzhiyi@cis.upenn.edu}}

\date{}

\begin{titlepage}
\thispagestyle{empty}

\maketitle


\begin{abstract} 
  \thispagestyle{empty}
  We provide a Polynomial Time Approximation Scheme (PTAS) for the Bayesian optimal multi-item multi-bidder auction problem under two conditions. First, bidders are independent, have additive valuations and are from the same population. Second, every bidder's value distributions of items are independent but not necessarily identical monotone hazard rate (MHR) distributions. For non-i.i.d.~bidders, we also provide a PTAS when the number of bidders is small. Prior to our work, even for a single bidder, only constant factor approximations are known. 
  
  Another appealing feature of our mechanism is the simple allocation rule. Indeed, the mechanism we use is either the second-price auction with reserve price on every item individually, or VCG allocation with a few outlying items that requires additional treatments. It is surprising that such simple allocation rules suffice to obtain nearly optimal revenue.
\end{abstract}
\end{titlepage}

\section{Introduction}

The multi-dimensional mechanism design problem has been widely studied in Economics, and recently in the theory of computation community. Consider a seller who has a limited supply of several distinguishable items and many interested bidders. The goal for the seller is to design an auction that will incentivize the bidders to truthfully report their private valuations and maximize her revenue. Unfortunately, optimal mechanism is not even well-defined in the worst-case analysis, as no truthful mechanism can be universally optimal for all possible valuation profiles. Economists have taken the Bayesian approach to cope with this impossibility, where the valuations of the bidders are assumed to be drawn from some publicly known distributions. Given such prior distributions, the optimal mechanism is defined as the one that maximizes the expected revenue among all (possibly randomized) truthful and individual rational mechanisms. In this paper, the notion of truthfulness we will focus on is Bayesian incentive compatibility (BIC), while we will also consider other notions of truthfulness such as incentive compatibility and deterministic truthfulness. Informally, a mechanism is BIC if each bidder maximizes her expected utility by truth-telling assuming other bidders are also truthful, where the expectation is over the randomness of the mechanisms and random realizations of other agents' valuations.

When there is only a single item for sale, the structure of the optimal
mechanism is very well-understood. Myerson \cite{myerson} provides an
elegant solution to the optimal single-item auction problem. However,
Myerson's result does not extend to the more general {\em multi-item}
setting. Following Myerson's work, a large body of research in Economics
has been devoted to extending his result to the multi-item setting (see
survey \cite{optimal:econ} and the references therein). 

The theory of computation community has also studied this problem during the past decade, with an eye on the computational efficiency of the mechanism. There has been lots of success in obtaining constant
factor approximations in various settings (e.g., \cite{CHK, CHMS, BGGM,
Alaei}). Lately, attention has been shifted to getting nearly optimal
revenue and such mechanisms have been proposed for several cases (e.g.,
\cite{DW, CDW12a, AlaeiFHHM12}). 

In a very recent paper \cite{CDW12b}, Cai et al.~consider a very general setting. In their setting, bidders are additive with arbitrary combinatorial feasibility constraints. They show how to design revenue-optimal auctions by reducing the revenue optimization to welfare optimization under the same constraints. Their algorithm has runtime polynomial in the total number of bidder types\footnote{More precisely, these algorithm is polynomial in $\sum_i |S_i|$, where $S_i$ is the support of the joint value distribution for bidder $i$.}. This is the natural description size for the problem if we allow items to have correlated values. However, when items are independent, the natural description is much more succinct. Making their algorithm inefficient (exponential in the input size). Moreover, to handle such a broad setting, their solution has to be relatively complicated, which might sometimes makes it hard to implement in reality. The above drawbacks motivate the research in this paper, that is, {\em designing simple, computational efficient, and nearly optimal auctions}.

\subsection{Main Results}

In this paper, we will focus on a very important and fundamental case: Bidders have independent and additive valuations, and items values are independent. Our goal is to obtain an algorithm whose runtime is polynomial in the succinct input-size and propose much simpler revenue-optimal auctions.

More concretely, let there be $m$ bidders\footnote{We will sometimes use $k$ to denote the number of bidders when this number is an absolute constant.} and $n$ heterogeneous items (unit-supply). Let there be no feasibility constraints on the allocations. We will assume the bidders' valuations are additive and the values are drawn from independent but not necessarily identical distributions subject to the standard monotone hazard rate (MHR) assumption. Roughly speaking, MHR distributions are those whose tails are ``thinner'' than exponential distributions. The formal definition of MHR is deferred to Section \ref{sec:prelim}. We want to efficiently find a mechanism whose expected revenue is optimal relative to any (possibly randomized) truthful and individual rational mechanism. Prior to our work, even the case of a single bidder is elusive in the presence of many independent but not necessarily identical items. 

Our main results are the following two theorems.

\begin{theorem}\label{thm:mainiid}
	Let there be $n$ heterogeneous items, $m$ additive bidders, and $\{\mathcal{F}_j\}_{j\in[n]}$ be a collection of independent but not necessarily identical MHR distributions. Suppose for each bidder $i$, her value for item $j$ is drawn independently from $\mathcal{F}_j$. Then, there is a Polynomial Time Approximation Scheme\footnote{Recall that a {\em Polynomial Time Approximation Scheme} (PTAS) is a family of algorithms $\{\mathcal{A}_{\epsilon}\}_{\epsilon}$, indexed by a parameter $\epsilon >0$, such that for every fixed $\epsilon>0$, $\mathcal{A}_{\epsilon}$ runs in poly-time. In particular, for any constant $\epsilon>0$, the PTAS constructs an auction whose expected revenue is a $(1+\epsilon)$ factor approximation to the optimal, in time polynomial in $n$ and $m$.} (PTAS) for computing the revenue-optimal truthful mechanism. 
\end{theorem}

In the above theorem, we consider the case when the bidders are from the same population. So any two bidders have the same value distributions for any particular item. This is a realistic assumption, as to tell which demographic group the bidder is from, the seller needs to collect lots of information, e.g., her occupation, income, marital status etc. which is usually infeasible in practice, especially when the number of bidders is huge. When there are only a handful of bidders, however, the seller might have enough knowledge to distinguish different bidders. We develop the following theorem to address this case.

\begin{theorem} \label{thm:main}
	Let there be $n$ heterogeneous items, $k$ additive bidders (consider $k$ as an absolute constant), and $\{\mathcal{F}_{ij}\}_{i\in[k],j\in[n]}$ be a collection of independent but not necessarily identical MHR Distributions. For any bidder $i$ and item $j$, her value for the item is drawn from $\mathcal{F}_{ij}$. There is a PTAS for computing the revenue-optimal BIC mechanism. 
\end{theorem}

Although the above theorems are stated only for BIC mechanism here, our techniques can be extended to other solution concepts as well, such as IC and deterministic truthfulness. We will elaborate these theorems in the corresponding sections and explain the results for various solution concepts.

Besides achieving nearly optimal revenue, our mechanisms in Theorem \ref{thm:mainiid} and Theorem \ref{thm:main} have an additional appealing feature of using very simple allocation rules. In fact, all of our mechanisms essentially has one of the two following simple forms: 1) Run a second price auction with reserve price on every item individually. 2) Use the VCG allocation with a threshold welfare whose role is similar to the reserve price, except for a few outlying items which we need to handle separately. It is surprising that such simple allocation rules can actually obtain nearly optimal revenue.


\subsection{Overview of Techniques}

First let us explain by example why the obvious attempt of running Myerson's auction on every item individually fails. Consider a single bidder and $n$ items whose values are i.i.d.~and uniformly drawn from $[0, 1]$. On the one hand, Myerson's optimal auction only gets $\frac{1}{4}$ revenue per item. On the other hand, if $n$ is large, the total value of the grand-bundle concentrated at $\frac{n}{2}$. So a simple grand-bundle-reserve-price auction (e.g., \cite{Arm99}) can get almost $\frac{n}{2}$ revenue. 

One might also argue that when the bidders' values are additive, the overall values will be concentrated and thus it is easy to find the optimal. But as items are non-i.i.d., we may not have such a concentration phenomena in some cases.\footnote{For instance, consider an item whose value is uniformly drawn from $[1/2, 1]$ and $n-1$ items whose values are i.i.d.~and uniformly drawn from $[0, \frac{1}{n^2}]$.}

Instead, our first technical contribution is by understanding the probabilistic structure to prove the following structural lemma which we will use heavily:

\bigskip
\noindent{\bf Partitioning Lemma (Informal).~} {\em Assuming MHR distributions, then we can partition the items into two sets, where the first set contains only a constant number of items, and the second set has many items but the social welfare of which highly concentrates.}

\bigskip

Based on this lemma, we manage to reduce the problem of finding nearly-optimal mechanisms for many independent items into two simpler sub-problems: Designing nearly-optimal mechanisms for a constant number of independent items, and designing nearly-optimal mechanisms when the total value of the items concentrates. The formal statement of the partition lemma and its proof will be given in Section \ref{sec:decomposition}.

\paragraph{Constant Number of Bidders} 

In this case, designing nearly-optimal mechanisms for the sub-problem with only a constant number of items is almost folklore and we sketch these mechanisms in Section \ref{sec:constant}. In order to handle the second sub-problem, we propose a novel mechanism that falls into the VCG family, which we shall introduce as the {\em reserve welfare mechanism} in Section \ref{sec:reservewelfare}. The reserve welfare mechanism allocates items to the bidders only if the social welfare exceeds a certain reserve welfare, in which case it will use the welfare-maximizing allocation. We show that with the proper pricing scheme, the reserve welfare mechanism is deterministically truthful and solves the welfare-concentrated case nearly optimally. The proof of Theorem \ref{thm:main} follows by combining these technical ingredients.

\paragraph{Many I.I.D.~Bidders}

The key observation in this case is that when the number of bidders is sufficiently large, simply running second price auction with a properly chosen reserve price for each item suffices to guarantee nearly optimal revenue. More concretely, inspired by Theorem 7 in \cite{CD}, we can argue that for any constant $\epsilon > 0$, if the number of bidders is larger than an absolute constant that only depends on $\epsilon$, then for every item there is a second price auction with reserve price that achieves revenue at least a $(1-\epsilon)$ fraction of the social welfare. In \cite{BK}, Bhalgat and Khanna have independently provided similar insights when there are sufficiently many i.i.d.~bidders. On the other hand, if the number of bidders is smaller than this absolute constant, then we can reduce the problem to Theorem \ref{thm:main}.

\subsection{Related Work}
\label{sec:related}

The theory of computation community has contributed many computational efficient solutions to various special cases of the multi-dimensional mechanism design problem. Chawla et.al~\cite{CHK} consider the case of a single unit-demand bidder, and propose an item pricing mechanism that achieves a constant factor approximation of the optimal. Their result is based on an elegant reduction to Myerson's optimal auction in the single-dimensional setting. For the same problem, Cai and Daskalakis \cite{CD} propose a PTAS for optimal item-pricing, thus close the constant approximation gap. 

In the multi-bidder setting, \cite{CHMS,BGGM,Alaei} provide efficient constant factor approximations for cases when the bidders are additive or unit-demand. More recently, near-optimal solutions have been obtained for several cases. Daskalakis and Weinberg \cite{DW} solve the case where there are few bidders with symmetric\footnote{See \cite{DW} for a formal definition of symmetric distributions. E.g., i.i.d.~distributions is symmetric, but general independent distributions are not.} items or symmetric bidders with few items. 

{For asymmetric distributions, Cai et al.~\cite{CDW12a} give the optimal solution to the many-bidder and many-item setting. Alaei et al.~\cite {AlaeiFHHM12} consider serving many copies of an item with a matroid feasibility constraint on which bidders can be served an item simultaneously, and obtain the optimal solution. In \cite{CDW12b}, Cai et al.~provide the optimal solution for a much more general setting where bidders are additive, and with (possibly) arbitrary feasibility constraints, by reducing the revenue optimization to welfare optimization. Their reduction provides a poly-time solution to the optimal mechanism design problem in all auction settings where welfare optimization can be solved efficiently. However, it is fragile to approximation, as the reduction requires an exact solution for the welfare optimization problem. In \cite{CDW12c}, the same group of authors show that even when the welfare optimization problem is only approximately solvable, they can still carry over the reduction while preserving the approximation factor. All of these
algorithms allow correlation among items, so the total number of 
bidder types is the natural input size.
However, for independent but not necessarily identical items, even when support size of every
value distribution is only $2$ and bidders are i.i.d., the total number of 
bidder types could still be as large as $2^{n}$, making their
algorithm highly inefficient in our setting. Nonetheless, for symmetric
items, Cai et al.~\cite{CDW12a} show how to reduce the ``effective
number'' of types by utilizing the symmetric structure of the items,
yielding mechanisms that are polynomial in both $n$ and $m$. But designing nearly optimal auctions for asymmetric items remains open prior to our work even for a single bidder.


Our result can be viewed as an improvement of \cite{DW} and a complement
to \cite{CDW12a}. Although the results are related, the techniques are orthogonal. The approaches in \cite{DW} and \cite{CDW12a} are
LP-based, and they use symmetry to reduce the size of the LP. We take a
different path. By understanding the probabilistic structure, we
argue that the social welfare of most of the items are highly
concentrated, and can be easily extracted by the seller using a modified
VCG mechanism. Further, for the other constant number of items, as there
    are only a small number of possible types, they can be easily handled by
previous results (e.g., \cite{DW}).}  


\section{Preliminaries}\label{sec:prelim}

\subsection{Model}

Formally, in an {\em multi-item auction}, a seller has $n$ heterogeneous items that she wants to auction to $m$ quasi-linear risk-neutral bidders. Each bidder $i$ has a private valuation profile $v_i = (v_{i1}, \dots, v_{in}) \in \R^n$, where $v_{ij}$ is bidder $i$'s value for item $j$. $v_i$ is sometimes referred to as the type of the bidder $i$.  We will assume the valuation function to be additive, that is, $v_i(S) = \sum_{j \in S} v_{ij}$ for any $S \subseteq [n]$. We will let $v_{-i}$ denote the type profile of every bidder except $i$.

A {\em mechanism} $M$ consists of two parts: An allocation rule $x(\cdot)$ and a payment rule $p(\cdot)$. 

The {\em allocation rule} $x(\cdot)$ maps a type profile $\bv$ to a feasible allocation $x(\bv) = \{x(\bv)_{ij}\}_{i \in [m], j \in [n]}$, where $x(\bv)_{ij}$ is the probability for bidder $i$ to receive item $j$ when the type profile is $\bv$. For deterministic mechanisms, we will let $x(\bv)_{ij}$ to be either $0$ or $1$. We will let $x(\bv)_i$ denote the $n$ dimensional vector $(x(\bv)_{i1}, \dots, x(\bv)_{in})$. 

The {\em payment rule} maps a type profile $\bv$ to a $m$-dimensional real vector $p(\bv) = (p_1(\bv), \dots, p_m(\bv))$, where $p_i(\bv)$ is the price charged to bidder $i$. 

Since the valuations are private information of the bidders, the mechanism needs to retrieve these information from the bidders, who may or may not manipulate the information. We will let $\bm{b} = (b_1, \dots, b_m)$ denote the {\em bids} of the bidders. Given the bids, the allocation $x(\bm{b})$, and the payments $p(\bm{b})$, we will assume the bidders are utility maximizers w.r.t.~the standard notion of {\em quasi-linear utility}:
$$u_i(v_i, x(\bm{b}), p(\bm{b})) = v_i \cdot x_i(\bm{b}) - p_i(\bm{b}) \enspace.$$

We will consider the {\em Bayesian setting}. Namely, we will assume that the valuations $v_{ij}$, $i \in [m]$ and $j \in [n]$, are drawn from some publicly known independent (but not necessarily identical) distributions $\calF_{ij}$. When bidders are from the same population, we will use $\mathcal{F}_{j}$ to denote the value distribution for item $j$, and omit subscript $i$. We will let $F_{ij}(x)$ and $f_{ij}(x)$ denote the cumulative distribution function and probability density function of $\calF_{ij}$ respectively. 

 
Next, we formally define how the distributions $\calF_{ij}$ are specified to the mechanism. We shall consider two different models. The first one is the {\em discrete explicit access} model, where the support of each $\calF_{ij}$ is discrete and explicitly given, and so is the probability of each value in the support being chosen. The second one is the {\em continuous oracle access} model, where the support of $\calF_{ij}$ could be continuous and even unbounded. In the latter, we will assume there is an oracle sampler such that each access to the oracle returns a random value drawn from $\calF_{ij}$. In the former case, the running time of our mechanisms shall be polynomial in the sum of the support sizes of $\calF_{ij}$ for all $i \in [m]$ and $j \in [n]$. In the latter case, both the running time and the number of accesses to the oracle of our mechanisms shall be polynomial in $m$ and $n$.

\subsection{Solution Concepts} 

We will consider the following standard game-theoretic solution concepts:

\begin{definition}
  A deterministic mechanism $M$ is {\em deterministically truthful} (DT), if truth-telling is a utility-maximizing strategy, i.e.,
  $$\forall b_{-i} : \quad v_i \in \arg\max_{b_i} \{v_i \cdot x_i(b_i, b_{-i}) - p(b_i, b_{-i}) \} \enspace.$$
\end{definition}

\begin{definition}
  A randomized mechanism $M$ is {\em truthful-in-expectation} or {\em incentive compatible} (IC) if truth-telling maximizes the expected utility, i.e.,
  $$\forall b_{-i} : \quad v_i \in \arg\max_{b_i} \{ \E[v_i \cdot x_i(b_i, b_{-i}) - p(b_i, b_{-i})] \} \enspace.$$
  where expectation is over random coin-flips of the mechanism.
\end{definition}

\begin{definition}\label{def:BIC}
  A (randomized) mechanism $M$ is {\em Bayesian-incentive-compatible} (BIC) if truth-telling maximizes the expected utility, i.e.,
  $$\forall b_{-i} : \quad v_i \in \arg\max_{b_i} \{ \E_{b_{-i}\sim\mathcal{F}_{-i}}[v_i \cdot x_i(b_i, b_{-i}) - p(b_i, b_{-i})] \} \enspace,$$
  where expectation is over random coin-flips of the mechanism and random realization of the valuations of other bidders.
\end{definition}

We will also consider the following relaxed notions of deterministic truthfulness.

\begin{definition}
  A deterministic mechanism $M$ is {\em $\eps$-deterministically truthful} ($\eps$-DT), if 
  $$\forall b_i, b_{-i} : \quad v_i \cdot x_i(b_i, b_{-i}) - p(b_i, b_{-i}) \le v_i \cdot x_i(v_i, b_{-i}) - p(v_i, b_{-i}) + \eps \enspace.$$
\end{definition}

The notions of $\epsilon$-IC and $\epsilon$-BIC are defined similarly.

\medskip 

Further, it is very important not to overcharge the bidders, especially
when we are aiming for revenue.

\begin{definition}
  A mechanism $M$ is {\em individually rational} (IR) if the utility of any bidder in any outcome is always non-negative, i.e., 
  $$\forall b_{-i} : \quad v_i \cdot x_i(v_i, b_{-i}) - p(v_i, b_{-i}) \ge 0 \enspace.$$
\end{definition}


The following {\em taxation principle} is a well known characterization for truthful mechanisms (e.g., see \cite{Dobzinski, DV}) which will be useful for our discussion. 

\begin{theorem}[Taxation Principle]
  A mechanism is DT/IC if and only if each bidder $i$ is presented a menu of bundles/lotteries of items such that the prices of the bundles/lotteries only depends on the other bidders' valuations $v_{-i}$, and bidder $i$ always gets one of the utility maximizing bundles/lotteries. 
\end{theorem}

In particular, if there is only one bidder, then such menus are fixed regardless of the reported value. So any DT/IC mechanism can be viewed as a bundle-pricing/lottery-pricing of the items.

\subsection{Extreme Value Theorem} 

Throughout this paper, we will consider distributions that have {\em monotone hazard rate} (MHR): 

\begin{definition}
  \label{def:MHR}
  A distribution $\calF$ has {\em monotone hazard rate} if $\frac{f(x)}{1 - F(x)}$ is non-decreasing in the support of $\calF$.
\end{definition}

The MHR distributions is a commonly studied family of distributions in
Economics and recently in the algorithmic game theory community. It includes familiar distributions such as the Normal,
Exponential, and Uniform distributions. Intuitively, a distribution has
monotone hazard rate if its tail is at most as large as that of an
Exponential distribution. We note that in our results the MHR assumption
can be replaced by the following assumption: There exists a constant $C$
such that the for each bidder $i$ and item $j$ the support of
$\calF_{ij}$ is an interval whose upper and lower bounds differs by at
most a $C$ multiplicative factor. In other words, our algorithms work
well as long as we have a rough idea on each bidder $i$'s value on each
item $j$. In fact, this is the alternative assumption we will use for the discrete explicit access model.

We will use the following {\em extreme value theorem} for MHR distributions developed in \cite{CD} as an important technical tool in our proofs. Readers are referred to \cite{CD} for the proof of the theorem.

\begin{theorem}[Extreme Value Theorem \cite{CD}] \label{lem:extremevalue}
  Suppose $X_1, \dots, X_n$ are a collection of independent (but not necessarily identically distributed) random variables whose distributions are MHR, and $f_{\max_i \{X_i\}}$ is the probability density function of the random variable $\max_i X_i$. Then, for all $\epsilon \in (0, \frac{1}{4})$, there exists some anchoring point $\beta$ such that $\Pr[\max_i X_i \ge \frac{\beta}{2}] \ge 1 - {1\over\sqrt{e}}$ and 
  $$\int^{+\infty}_{2 \beta \log(\frac{1}{\epsilon})} t \cdot f_{\max_i \{X_i\}}(t) dt \le 36 \beta \epsilon \log \left(\frac{1}{\epsilon}\right) \enspace.$$
  Moreover, $\beta$ is efficiently computable from the distributions of the $X_i$'s.
\end{theorem}

Based on the above extreme value theorem and an additional probabilistic
argument, we will show in Section \ref{sec:decomposition} that the social
welfare of some carefully chosen subset of items highly concentrates. In
particular, we will consider the followings notion of concentration.

\begin{definition}
  A random variable $X$ is {\em $(\epsilon, \delta)$-concentrated} if $X \in (1-\epsilon, 1+\epsilon) \E[X]$ with probability at least $1-\delta$.
\end{definition}

\section{Nearly Optimal Mechanism for Constant Number of Bidders}  
\label{sec:decomposition}

{In this section, we will consider the case when there are only 
$k$ bidders, where $k$ is an absolute constant.\footnote{We note that our mechanisms can be
extended to the case of $O(\log^c n)$ bidders for sufficiently small
constant $c$ via almost identical proofs, where $c$ depends on the
solution concept. However, we feel such extension is not very
insightful. So we will only present the case for a constant number of
bidders in this extended abstract for the sake of presentation.} We will
prove that for various solution concepts, the problem of finding
revenue-optimal truthful mechanisms can be solved under a unified
framework. Formally, our results can be summarized as the following
theorem.

\bigskip
\noindent\begin{minipage}{\textwidth}
    \begin{theorem}[Thm.~\ref{thm:main} elaborated] \label{thm:mainrestate}
    Suppose the number of bidders is a constant. Then, there is a PTAS (polynomial in $n$) for finding revenue-optimal mechanisms among (all settings require IR): 
    \begin{itemize} 
        \setlength{\partopsep}{0pt}
        \setlength{\topsep}{0pt}
        \setlength{\parsep}{0pt} 
        \setlength{\itemsep}{0pt}
        \item IC/BIC mechanisms with discrete explicit access.
        \item DT mechanisms with discrete explicit
            access.\footnotemark[1]
        \item DT/IC mechanisms with continuous oracle access for a single bidder.
        \item BIC mechanisms with continuous oracle access.
        \item DT/IC mechanisms with continuous oracle access.\footnotemark[1]
    \end{itemize}
    \footnotetext{\footnotemark[1]~ Our mechanisms in these cases are only
    $\epsilon$-deterministically truthful and $\epsilon$-IC.}
\end{theorem}
\end{minipage}
\bigskip

The general proof strategy of Theorem \ref{thm:mainrestate} is to reduce 
the problem of designing almost optimal mechanisms for 
the multi-item auction problem into two easier sub-problems (assuming
MHR distributions). More precisely, we will prove that if there are
PTAS for the special cases in the next two lemmas, then it is
possible to combine the nearly optimal mechanisms for these two cases to
derive the PTAS in Theorem \ref{thm:mainrestate}.

\begin{lemma}[Few-Item Case] \label{lem:fewitem}
    Theorem \ref{thm:mainrestate} holds if both the number of items and the number of bidders are constants.
\end{lemma}

\begin{lemma}[Concentrated Case] \label{lem:concentrated}
    Suppose the optimal social welfare is $(\epsilon, \delta)$-concentrated, and the number of bidders is a constant. Then, there is a polynomial-time and deterministically truthful mechanism whose expected revenue is at least $(1 - f(\epsilon, \delta))$ fraction of the expected optimal social welfare, where $f(\epsilon, \delta)$ goes to zero as $\epsilon$ and $\delta$ goes to zero.
\end{lemma}

At this point, we will focus on how to combine the mechanisms obtained
from the above lemmas to derive the proof of Theorem \ref{thm:mainrestate}. The
proofs of Lemma \ref{lem:fewitem} and Lemma \ref{lem:concentrated} are deferred to
Section \ref{sec:constant} and Section \ref{sec:reservewelfare} respectively.

\paragraph{Proof Outline of Theorem \ref{thm:mainrestate}}
 
Before getting into the technical details, let us first sketch the
road-map of our proof. First of all, we notice that under the MHR
assumption, it is easy to achieve expected revenue that is at least a
constant fraction of the expected social welfare. This follows easily from
previous work (e.g.~see \cite{BBHM}) and we will formally state it as
Lemma \ref{lem:constantfactor}. By this result, we know that we can throw
away items whose contribution to the expected social welfare is tiny
without overhurting the optimal revenue. Next, we proceeds by proving a
structural result saying that we can partition the items into three
groups: a small group of items with large variance, which we shall
handle with the mechanism from Lemma \ref{lem:fewitem}
(Section \ref{sec:constant});\footnote{The philosophy of selecting a
few distinguished items to reduce the size of the problem and solve it
nearly optimally
may looks similar to that of the $k$-lookahead auction (e.g.~see
\cite{Ronen, DFK}), where we choose a few distinguished bidders and
design nearly optimal mechanism for them (based on the bids of the other
bidders). However, there are a few crucial differences. First, in our
approach the small set of items are chosen without knowing the bids
while in the $k$-lookahead auction the set of bidders are chosen based
on the bids. Further, in our approach we also derive good revenue from
the rest of the items while the $k$-lookahead auction never derive
revenue directly from the rest of the bidders. Finally, as a result of
the previous point, our approach admits nearly optimal revenue while the
$k$-lookahead auction only guarantees constant-factor approximation so
far.}
a group of items whose contribution to the social welfare concentrates,
which we will handle with the mechanism from Lemma \ref{lem:concentrated}
(Section \ref{sec:reservewelfare}); and finally a group of items whose total
contributions to the expected social welfare is tiny, which we will
simply ignore (never allocate them to any bidder). This result is
formally stated and proved in Lemma \ref{lem:partition}. At last, in order to
show this approach is a PTAS for the multi-item auction problem, we need
to show that the optimal revenue of the problem is upper bounded by the
optimal revenue when only a subset of items present (the items with
large variance) plus the expected social welfare of the remaining items
(the concentrated group of items). Indeed, we will prove this claim as
Lemma \ref{lem:2}. The complete proof of Theorem \ref{thm:mainrestate} is given in
Appendix \ref{app:decomposition}.

\medskip

Next, let us formally state and prove the technical lemmas mentioned in the proof sketch. The following lemma is folklore from previous work. 

\begin{lemma}[E.g., Corollary 3.7 of \cite{BGGM}] \label{lem:constantfactor}
    For any multi-item auction with MHR bidders, the optimal expected
    revenue is at least a constant fraction of the expected social welfare.
\end{lemma}

Now let us consider the partition lemma. For presentation purpose, we will
only show a weaker version of the partition lemma, under the additional
assumption that the upper and lower bounds of the value range of $X_j =
\max_i v_{ij}$ only differ by at most a constant factor $c$ for every
item $j$. Note that Lemma \ref{lem:extremevalue} implies that at least $1 -
O(\epsilon\log{1\over\epsilon})$ fraction of the contribution to
$\E[X_j]$ comes from a range whose upper and lower bounds
differ by at most an $O(\frac{1}{\epsilon} \log(\frac{1}{\epsilon}))$
factor. It is easy to see that by choosing $c = \Theta(\frac{1}{\epsilon}
\log(\frac{1}{\epsilon}))$ and taking into account the fact that the
contribution outside the range is tiny, we can prove the 
partition lemma without the additional assumption. We omit the details here.

\begin{lemma} \label{lem:partition}
    Suppose $X_1, \dots, X_n$ are $n$ non-negative independent
    random variables, where $[\alpha_j, \beta_j]$ is the range of $X_j$, 
    $1 \le j \le n$, such that $c = \max_j \frac{\beta_j}{\alpha_j}$ is a
    constant. Suppose $\epsilon > 0$ and $\frac{1}{8} > \delta > 0$ are
    small constants. Then, we can partition $X_1, \dots, X_n$ into three
    groups $R$, $S$, and $T$ in polynomial time, such that: 
    \begin{enumerate}
        \setlength{\partopsep}{0pt}
        \setlength{\topsep}{0pt}
        \setlength{\parsep}{0pt} 
        \setlength{\itemsep}{0pt}
        \item The size of $R$ is small: $\abs{R} \le \frac{16
            c^2}{\epsilon^3} \ln \left(\frac{2}{\delta}\right)$. 
        \item The sum in $S$, $\sum_{X_j \in S} X_j$,
            is $(\epsilon,\delta)$-concentrated.
        \item The contribution from group $T$ is tiny: $\sum_{X_j \in T}
            \E[X_j] \le \epsilon \sum_{j=1}^n \E[X_j]$.
    \end{enumerate}
\end{lemma}

\begin{proof}
    Let $s = \sum_{j=1}^n \E[X_j]$ be the sum of the expectation of
    these random variables. Note that $s$ can be estimated up to a
    constant factor in polynomial time and such estimated value is
    sufficient for our purpose. For the sake of presentation, we will
    assume that we known the value of $s$.
    
    We will first partition the random variables into $\Theta(\log n)$
    buckets $B_1, B_2, \dots, B_{\log n + \log(\frac{2}{\epsilon})}$
    according to their expectations. If
    $\E[X_j] \in [\frac{s}{2^\ell}, \frac{s}{2^{\ell-1}}]$ for $1 \le \ell
    \le \log n + \log(\frac{2}{\epsilon})$, then we put $X_j$ into
    bucket $B_\ell$. If $\E[X_j] \le \frac{\epsilon s}{2n}$, then its
    contribution to the social welfare is negligible and we will put
    $X_j$ into $T$. 

    Briefly speaking, we will proceed as follows. First pick a small
    threshold index $\ell^*$ (the value of $\ell^*$ will be defined
    later); then for each bucket $B_\ell$ such that $\ell \le \ell^*$,
    we put all random variables in $B_\ell$ into $R$; for each bucket
    $B_\ell$ such that $\ell > \ell^*$, we will show that either
    $\sum_{X_j \in B_\ell} X_j$ concentrates with high probability,
    or the contribution of $\sum_{X_j \in B_\ell} \E[X_j]$ is tiny. 
    In the former case, we will put the variables in $B_\ell$ into $S$;
    in the latter case, we will put the variables in $B_\ell$ into $T$. 

    More precisely, for each bucket $B_\ell$ where $\ell > \ell^*$, if
    $|B_\ell| \ge \frac{2c^2}{\epsilon^2} (\ln(\frac{2}{\delta}) + \ell
    - \ell^*)$, we will put these random variables into $S$. Note that
    for every $X_j \in B_\ell$, we have that $\beta_j \le c \alpha_j \le
    c \E[X_j] \le \frac{c \cdot s}{2^{\ell-1}}$ and $\E[X_j] \ge
    \frac{s}{2^\ell}$. By Chernoff-H\"oeffding bound, we get that
    $$\Pr\left[\abs{\sum_{X_j \in B_\ell} X_j - \sum_{X_j \in B_\ell} \E[X_j]} > \epsilon \sum_{X_j \in B_\ell} \E[X_j]\right] \le 2 \exp\left(-\ln \left(\frac{2}{\delta}\right) - \ell + \ell^* \right) = \delta \, \exp(- \ell + \ell^*) \enspace.$$
    
    Now consider all the buckets that we put into $S$. By union
    bound, the probability that the sum of the random variables in any
    of these buckets does not concentrate is at most $\sum_{\ell >
    \ell^*} \delta \exp(-\ell + \ell^*) < \delta$. Thus, we have proved
    that $S$ satisfies the desired property in the lemma.

    If $|B_\ell| < \frac{2c^*}{\epsilon^2} (\ln(\frac{2}{\delta}) + \ell
    - \ell^*)$, we shall put all variables in $B_\ell$ into $T$. In this
    case, we have
    $\sum_{X_j \in B_\ell} \E[X_j] \le \frac{s}{2^{\ell-1}}
    \frac{2c^2}{\epsilon^2} \left(\ln\left(\frac{2}{\delta}\right) +
    \ell - \ell^*\right)$.
    Therefore, the sum of the expected values of the random variables in
    $T$ is at most 
    \begin{eqnarray*}
        \sum_{\ell > \ell^*} \frac{s}{2^{\ell-1}} \frac{2c^2}{\epsilon^2} \left( \ln \left( \frac{2}{\delta} \right) + \ell - \ell^* \right) + \sum_{X_j: \E[X_j] \le \frac{\epsilon s}{2 n}} \frac{\epsilon s}{2 n} & \le & \sum_{\ell > \ell^*} \frac{s}{2^{\ell-1}} \frac{2c^2}{\epsilon^2} \ln \left(\frac{2}{\delta}\right) + \sum_{\ell > \ell^*} \frac{s}{2^{\ell-1}} \frac{2c^2}{\epsilon^2} (\ell - \ell^*) + \frac{\epsilon s}{2} \\
        & = & \frac{s}{2^{\ell^*-1}} \frac{2c^2}{\epsilon^2} \ln \left(\frac{2}{\delta}\right) + \frac{s}{2^{\ell^*-2}} \frac{2c^2}{\epsilon^2} + \frac{\epsilon s}{2} \enspace.
    \end{eqnarray*}

    In order to guarantee that $\sum_{X_j \in T} \E[X_j] \le \epsilon
    s$, it suffices to choose $\ell^*$ such that 
    $$\frac{1}{2^{\ell^*-1}}
    \frac{2c^2}{\epsilon^2} \ln\left(\frac{2}{\delta}\right) +
    \frac{1}{2^{\ell^*-2}} \frac{2c^2}{\epsilon^2} \le \frac{\epsilon}{2} \enspace.$$
    
    Note that $\delta < \frac{1}{2 e^2}$ implies that 
    $$\frac{1}{2^{\ell^*-1}} \frac{2c^2}{\epsilon^2} \ln
    \left(\frac{2}{\delta} \right) + \frac{1}{2^{\ell^*-2}}
    \frac{2c^2}{\epsilon^2} \le \frac{1}{2^{\ell^*-2}}
    \frac{2c^2}{\epsilon^2} \ln \left(\frac{2}{\delta}\right) \enspace.$$
    
    We shall let $\ell^* = \log \left(\frac{16 c^2}{\epsilon^3}
    \ln\left(\frac{2}{\delta}\right)\right)$ and conclude that $T$
    satisfies the claimed property.

    Finally, we note that for any $X_j \in R$, we have $\E[X_j] \ge
    \frac{s}{2^{\ell^*}}$. So by $\sum_{X_j \in R} \E[X_j] \le s$ we get
    that the size of $R$ is at most $2^{\ell^*} = \frac{16
    c^2}{\epsilon^3} \ln \left(\frac{2}{\delta}\right)$.
\end{proof}

At last, we will show that by decomposing the problem into two
sub-problems we do not hurt the optimal revenue by too much. Concretely,
for any $S \subseteq [n]$, we let $\opt^{\textrm{DT}}(S)$,
$\opt^{\textrm{IC}}(S)$, and $\opt^{\textrm{BIC}}(S)$ denote the optimal
revenue by deterministically truthful/IC/BIC mechanisms respectively
when only the items in $S$ are available on the market (value
distributions are the same). We have 

\begin{lemma} \label{lem:2}
    For any $S \subseteq [n]$, we have 
    $$\opt^{\textrm{truthful}} ([n]) \le \opt^{\textrm{truthful}}(S) + \sum_{j \notin S} \E[\max_i v_{ij}] \enspace,$$ 
    where $\textrm{truthful}$ can be instantiated with deterministically truthful (DT), or IC, or BIC.
\end{lemma}

\begin{proof}
    Suppose $M$ is the truthful (under the instantiated solution
    concept) mechanism that achieves optimal  
    revenue. Let us construct a truthful mechanism $M_S$ for the market
    when only the items in $S$ is presented. The revenue of the
    mechanism shall be at least
    $\opt^{\textrm{truthful}}([n]) - \sum_{j \notin S} \E[\max_i v_{ij}]$: 
    \begin{enumerate}
        \setlength{\partopsep}{0pt}
        \setlength{\topsep}{0pt}
        \setlength{\parsep}{0pt} 
        \setlength{\itemsep}{0pt}
        \item Let bidders submit their bids $\bb_{1, S}, \dots, \bb_{k,
            S}$.
        \item Sample values $\bv_{1, -S} \sim \calF_{1, -S},
            \dots, \bv_{k, -S} \sim \calF_{k, -S}$ for items not in $S$.
        \item Run $M$ on bids $(\bb_{1, S}, \bv_{1, -S}), \dots,
            (\bb_{k, S}, \bv_{k, -S})$. Let $\bS$ and $\bp$ denote the
            allocation and prices.
        \item Give bidder $i$ the items in $S_i \cap S$ and charge
            her $p_i - \sum_{j \in S_i \setminus S} v_{ij}$.
    \end{enumerate}

    First, let us analyze the revenue achieved by $M_S$ assuming the
    bidders bid truthfully: $\bb_{i, S} = \bv_{i, S}$ for $1 \le i \le k$. 
    The revenue by $M_S$ is $\sum_{i=1}^k \E\left[p_i - \sum_{j \in S_i
    \setminus S} v_{ij}\right]$. By linerity of expectation, this can be
    divided into two parts: $\sum_{i=1}^k \E[p_i] - \sum_{i=1}^k
    \E[\sum_{j \in S_j \setminus S} v_{ij}]$. 
    The first part $\sum_{i=1}^k \E[p_i]$ is the
    expected revenue $\opt^{\textrm{truthful}}([n])$ achieved by $M$
    The second part $\sum_{i=1}^k \E \left[ \sum_{j \in S_i \setminus S}
    v_{ij} \right]$ is social welfare from items outside $S$, which is
    upper bounded by the optimal social welfare . Note that the latter
    part is upper bounded by the $\sum_{j \notin S} \E\left[\max_i
    v_{ij}\right]$. Therefore, the revenue by $M_S$ is at least 
    $$\textstyle \opt^{\textrm{truthful}}([n]) - \sum_{j \notin S} 
    \E \left[ \max_i v_{ij} \right] \enspace.$$

    Now let us explain why mechanism $M_S$ is indeed truthful with
    respect to the corresponding solution concepts. Note that
    bidder $i$'s utility is 
    $$\sum_{j \in S_i \cap S} v_{ij} - (p_i - \sum_{j \in S_i \setminus S} v_{ij}) = \sum_{j \in S_i} v_{ij} - p_i \enspace,$$
    which is exactly the utility of a virtual bidder whose values
    are $\bv_i$ and bids $(\bb_{i, S}, \bv_{i, -S})$. So if $M$ is
    IC/BIC, then $M_S$ is also IC/BIC.
    Finally, if $M$ is a deterministically truthful mechanism 
    then $M_S$ is uniformly truthful. We further note that 
    there is no performance gap between optimal uniformly truthful
    mechanisms and optimal deterministically truthful mechanisms in the
    Bayesian setting. So Lemma \ref{lem:2} follows.
\end{proof}

\subsection{Nearly Optimal Mechanism for Constant Number of Items and
bidders} \label{sec:constant}

The mechanisms for constant number of items and constant number of bidders mostly follow directly from previous work. The general approach is to brute-force search with the hope that the search space would be small since both the number of items and the number of bidders are small. However, the strategy spaces for mechanism design problems are typically infinite. Hence, appropriate discretization is needed in order to reduce the size of the search space. We will briefly describe these mechanisms and thus prove Lemma \ref{lem:fewitem} in the Appendix \ref{app:constlp} for self-containness.

\subsection{Nearly Optimal Mechanism When Social Welfare Concentrates} 
\label{sec:reservewelfare}

In this section, we will prove Lemma \ref{lem:concentrated} by demonstrating
how to design nearly optimal mechanisms, when the social
welfare concentrates near its expectation as the bidders' values are drawn
from the corresponding distributions. 

\subsubsection{Single-Bidder Case}

As a warm-up, let us first consider the single-bidder case. This case is
quite straight-forward. We note that a {\em grand-bundle-reserve-price
auction} (e.g.~see Armstrong \cite{Arm99}) shall suffice. More precisely,
the auction will offer the bidder the grand bundle with a
take-it-or-leave-it price 
$$r^* = (1 - \epsilon) \E[\sum_j v_j] \enspace.$$

If the bidder values the grand bundle above $r^*$, she will
take the grand bundle and pay $r^*$; no item is allocated otherwise and
the bidder pays nothing. The proof of the next theorem follows
straightforwardly from the definition of the mechanism and $(\epsilon,
\delta)$-concentrated. So we will omit the tedious details.

\begin{theorem} \label{thm:singlebidder}
    The grand-bundle-reserve-price auction is deterministically
    truthful, individually rational, 
    and its expected revenue is at least $(1 - \epsilon)(1 - \delta)
    \E[\sum_j v_j]$ if the social welfare is $(\epsilon,
    \delta)$-concentrated.
\end{theorem}

\subsubsection{Constant Number of Bidders}

Now we show a similar result for multiple bidders. As a natural
first attempt, it might be tempting to think there exists a {\em
reserve-revenue mechanism} with reserve revenue $r^* = (1 - \epsilon)
\E[\sum_j \max_i v_{ij}]$ such that the mechanism offers the grand-bundle
to all the bidders at a reserve price $r^*$ and let the bidders discuss
 and decide whether to accept this offer and how to
share the items and the costs if they decide to accept. 
Of course, the last step in the above procedure is not well-defined.
The hope is that there is a truthful way for the bidders to come to a
consensus of accepting the offer whenever the optimal social welfare is
greater than $r^*$, since in such cases the bidders as a whole has
positive surplus when buying the grand-bundle at the reserve price
$r^*$. 
It is easy to see that this mechanism (if implementable) achieves a
revenue of at least $(1-\epsilon)(1-\delta) \E[\sum_j \max_i v_{ij}]$.
Unfortunately, we show that such mechanisms cannot be implemented in a
truthful and IR manner. We will defer the discussion
of this impossibility result to Appendix \ref{app:reserverevenue}.

\paragraph{Reserve-Welfare Mechanism}

In order to handle the multiple-bidder case, we will propose a novel
mechanism in the VCG-family, which we shall refer to as the {\em
reserve-welfare mechanism}. 
The idea is the following: it might be too aggressive to ask for a certain
reserve revenue whenever the social welfare is above this reserve
revenue; but it suffices to aim for the weaker goal of getting good
revenue only when the social welfare is closed to its expectation
because the social welfare concentrates by our assumption.
Concretely, the {\em reserve-welfare mechanism} is defined in
Figure \ref{fig:reservewelfare}.

\begin{figure*} 
    \centering
    \noindent
    \fbox{
    \begin{minipage}{.9\textwidth}
        \begin{enumerate}
            \item Let $\hat{s} = (1-\epsilon) \E \left[ \sum_j \max_i
                v_{ij} \right]$ be the reserve welfare.
            \item If the optimal social welfare according to the bids,
                $\sum_j \max_i b_{ij}$, is smaller the reserve welfare
                $\hat{s}$, then no item is allocated and the bidders pays
                nothing.
            \item Otherwise, allocate items according to an allocation $\bS
                = (S_1, \dots, S_k)$ that maximizes the social welfare.
            \item Charge bidder $i$ price $p_i = \hat{s} - \sum_{\ell \ne i}
                \sum_{j \in S_\ell} b_{\ell j}$.
        \end{enumerate}
    \end{minipage}}
    \caption{The reserve-welfare mechanism} \label{fig:reservewelfare}
\end{figure*}

Notice when $k=1$, this mechanism indeed becomes the
grand-bundle-reserve-price auction. So the reserve-welfare mechanism can
be viewed as a generalization of the grand-bundle-reserve-price auction.
We shall prove that this mechanism satisfies the desired properties.

\begin{theorem} \label{thm:multibidder}
    The reserve-welfare mechanism is deterministically truthful,
    individually rational, and its expected revenue is at least $(1 - k
    \epsilon- k \delta) \E[\sum_j \max_i v_{ij}]$ if the social
    welfare is $(\epsilon, \delta)$-concentrated for constants
    $\frac{1}{3} > \epsilon > 0$ and $1 > \delta > 0$.
\end{theorem}

Briefly speaking, the proof goes as follows. By our definition of the
payments, each bidder pays almost up to her value on the subset she gets
when if social welfare is near the reserve welfare $\hat{s}$. Further,
the social welfare will be near the reserve welfare $\hat{s}$ almost
for sure by our choice of $\hat{s}$ and that the social welfare is
$(\epsilon, \delta)$-concentrated. The only catch is the prices in the
reserve-welfare mechanism might be negative when the values of the
bidders are very large. We manage to show that the contribution of the
negative prices can be bounded as well. So the expected revenue almost matches the expected social welfare. Below let us present the formal argument.

\begin{proof}{\em (Theorem \ref{thm:multibidder})~}
    If we omit step 2 and always allocate items according to the
    social-welfare-maximizing allocation, then the mechanism falls into
    the VCG family except that we are using the reserve welfare
    $\hat{s}$ as our pivot instead of the Clarke pivot. So this variance
    of the reserve-welfare mechanism is deterministically truthful. Yet
    it is not individually rational. The reason of doing step 2 is
    exactly to fix the individual rationality.

    Formally, for each bidder $i$, suppose her true valuations are
    $\bv_i$ and she bids $\bb_i$. If reporting $\bv_{i}$ the items will
    not be allocated, then she should not lie and get the items
    allocated, since in the former case, her utility is $0$, while in
    the latter case her utility is negative. Now assuming the items are
    allocated, her utility is 
    $$\sum_{j \in S_i} v_{ij} - p_i = \sum_{j \in S_i} v_{ij} + \sum_{\ell \ne i} \sum_{j \in S_k} b_{\ell j} - \hat{s} \enspace.$$
    
    Note that the mechanism chooses the allocation
    that maximizes $\sum_i \sum_{j \in S_i} b_{ij}$. So by reporting her
    value truthfully the bidder maximizes her utility. Thus, the
    mechanism is deterministically truthful. Moreover,
    step 2 guarantees that allocation will be made only if $\sum_i
    \sum_{j \in S_i} b_{ij} \ge \hat{s}$. Therefore, the mechanism is
    individually rational.

    Finally, let us analyze the revenue achieved by the reserve-welfare
    mechanism. We will let 
    $$\textstyle s^* = \E\left[ \sum_j \max_i v_{ij} \right]$$
    denote the optimal expected social welfare, and recall that $\hat{s}
    = (1 - \epsilon) s^*$ is the reserve welfare. Assuming the bidders
    bid truthfully, the revenue is zero if the social welfare is less
    than $\hat{s}$, and is the following otherwise:
    \begin{eqnarray*}
      \sum_i p_i & = & \sum_i \left( \hat{s} - \sum_{\ell \ne i} \sum_{j \in S_\ell} v_{\ell j} \right) \\
      & = & k \hat{s} - \sum_\ell \sum_{j \in S_\ell} (k-1) v_{\ell j} \\
      & = & k \hat{s} - (k-1) \sum_{j \in [n]} \max_\ell v_{\ell j} \enspace.
    \end{eqnarray*}

    By our assumption that the social welfare is $(\epsilon,
    \delta)$-concentrated, the expected revenue is at least
    \begin{equation} \label{eq:app2}
        (1 - \delta) (1 - \epsilon) k s^* - (1 - \delta) (k - 1) \cdot \E \left[ \sum_{j \in [n]} \max_\ell v_{\ell j} \,|\, \sum_{j \in [n]} \max_\ell v_{\ell j} \ge \hat{s} \right] \enspace.
    \end{equation}

    Note that
    \begin{eqnarray}
        s^* & \ge & \Pr \left[ \sum_{j \in [n]} \max_\ell v_{\ell j} \ge \hat{s} \right] \E \left[ \sum_{j \in [n]} \max_\ell v_{\ell j} \,|\, \sum_{j \in [n]} \max_\ell v_{\ell j} \ge \hat{s} \right] \notag \\
        & \ge & (1 - \delta) \E \left[ \sum_j \max_\ell v_{\ell j} \,|\, \sum_j \max_\ell v_{\ell j} \ge \hat{s} \right] \enspace. \label{eq:app3}
    \end{eqnarray}

    Combining \eqref{eq:app2} and \eqref{eq:app3} we get that the expected
    revenue of the reserve-welfare mechanism is at least 
    $$\textstyle (1 - \delta) (1 - \epsilon) k s^* - (k - 1) s^* \ge (1
    - k \delta - k \epsilon) s^* \enspace.$$
    
    This proves the desired revenue guarantee.
\end{proof}

Since the number of bidders is an absolute constant, Theorem \ref{thm:multibidder} implies Lemma \ref{lem:concentrated}.

\section{Many Bidders From the Same Population}

As a natural restriction of the general multi-item auction probelm, we
will consider multi-item auctions with arbitrary number of items and
bidders under the assumption that they are from the same population. Formally, for
every item $j$, the value distributions $\calF_{ij}$ are identical for
every bidder $i$. In this case, we manage to design
nearly-optimal mechanisms based on our results for the few-bidder case.

\bigskip

\noindent\begin{minipage}{\textwidth}
\begin{theorem}[Thm.~\ref{thm:mainiid} elaborated] \label{thm:extension}
    Suppose the bidders are from the same population, then there
    is a PTAS (polynomial in both $n$ and $m$) for finding
    revenue-optimal mechanisms among (all settings require the mechanism
    to be IR): 
    \begin{itemize} 
        \setlength{\partopsep}{0pt}
        \setlength{\topsep}{0pt}
        \setlength{\parsep}{0pt} 
        \setlength{\itemsep}{0pt}
        \item IC/BIC mechanisms with discrete explicit access.
        \item DT mechanisms with discrete explicit access.\footnotemark[1]
        \item BIC mechanisms with continuous oracle access.
        \item DT/IC mechanisms with continuous oracle access.\footnotemark[1]
    \end{itemize}
    \footnotetext{\footnotemark[1]~ Our mechanisms in these cases are only
    $\epsilon$-deterministically truthful and $\epsilon$-IC.}
\end{theorem}
\end{minipage}

\bigskip

We will need the following lemma for i.i.d.~MHR distributions following Theorem 7 in 
Cai and Daskalakis \cite{CD} and the proof therein.

\begin{lemma}
    \label{lem:extremevalue2}
    Suppose $v_1, \dots, v_k$ are i.i.d.~according to a MHR distribution, and $k \ge (12/\epsilon)^{12/\epsilon}$. Then there is a threshold $r^*$ such that 
    $$\Pr \left[ \max_i v_i \ge r^* \right] \cdot r^* \ge (1 - \epsilon) \E \left[ \max_i v_i \right] \enspace.$$
    Moreover, one can efficiently find such a threshold $r^*$ in polynomial time.
\end{lemma}

Roughly speaking, Cai and Daskalakis managed to improve their extreme
value theorem when the bidders are i.i.d.~so that consider the
expectation of the random variable $\max_i v_i$, we only need to focus
on the contribution from a small interval whose upper and lower bounds
only differ by a $(1 + \epsilon)$ factor. As a simple corollary of this
stronger extreme value theorem, we have the above lemma.

Equipped with this lemma, we are now ready to solve the case of
arbitrary number of i.i.d.~bidders.

\begin{proof}{\em (Theorem \ref{thm:extension})~}
    Note that for each item $j$, the bidders' valuations for this item
    $v_{1j}, \dots, v_{kj}$ are i.i.d.~random variable according to a
    MHR distribution. Therefore, if the number of bidders $k$ is greater
    than $(12 / \epsilon)^{12 / \epsilon}$, then by
    Lemma \ref{lem:extremevalue2}, we can find in polynomial time a threshold
    $r^*_j$ for each item $j$ such that 
    $$\textstyle \Pr[\max_i v_{ij} \ge r^*_j] \cdot r^*_j \ge (1 -
    \epsilon) \E[\max_i v_{ij}] \enspace.$$

    Therefore, if we run the second price auction with reserve prices $r^*_j$ for
    each item $j$, then the expected revenue is at least 
    $$\textstyle \sum_j \Pr[\max_i v_{ij} \ge r^*_j] \cdot r^*_j \ge (1
    - \epsilon) \sum_j \E[\max_i v_{ij}] \enspace.$$

    Note that the right-hand-side of the above inequality is the optimal
    expected social welfare and therefore is an upper bound on the
    optimal revenue. So in the case when there are at least $(12 /
    \epsilon)^{12 / \epsilon}$ bidders, a simple
    reserve-price auction suffices to obtain a $(1 - \epsilon)$ fraction
    of the optimal revenue. Note that this mechanism is deterministic
    truthful and thus satisfies all our definitions of truthfulness.

    So it suffices to solve the case when the number of bidders are
    smaller than $(12 / \epsilon)^{12/ \epsilon}$. But this falls into
    the case of constant number of bidders for any constant $\epsilon >
    0$. So we could use the mechanism in Theorem \ref{thm:main} to solve the
    few-bidder case. In sum, we have proved the theorem.
\end{proof}

\bibliographystyle{plain}
\bibliography{multiitem}

\appendix

\section{Omitted Proofs in Section \ref{sec:decomposition}}
\label{app:decomposition}

\begin{proof}{\em (Theorem \ref{thm:mainrestate})~}
First, we will use Lemma \ref{lem:extremevalue} to truncate the random
variables $X_{j}=\max_{i} v_{ij}$ and get $\hat{X}_{j}$, so that every
$\hat{X}_{j}$ lies in an interval where upper bound and lower bound are only
$\left({1\over\epsilon}\log({1\over\epsilon})\right)$ factor away and
$\E[\hat{X}_{i}] \ge \left(1-O(\epsilon \log({1\over\epsilon}))\right)\E[X_{i}]$. 

By Lemma \ref{lem:partition}, we partition $[n]$ into three sets $R$, $S$ and
$T$ using the same $\epsilon$. Let $S_{1}$, $S_{2}$ and $S_{3}$ be the
sets of items whose max value are in $R$, $S$ and $T$, respectively.
Then, the size of $S_{1}$ is a constant that only depends on $\epsilon$
and the following its true. 
\begin{eqnarray}
    \sum_{j \in S_{3}} \E[\max_{i} v_{ij}] & \le & {\sum_{\hat{X}_i \in
    T}\E[\hat{X}_i]\over 1-O\left(\epsilon
    \log({1\over\epsilon})\right)} \notag \\
    & \le & {\epsilon \sum_{i=1}^n
    \E[X_i]\over 1-O\left(\epsilon \log({1\over\epsilon})\right)} \notag \\
    & \le &
    {O(\epsilon) \opt([n])\over 1-O\left(\epsilon
    \log({1\over\epsilon})\right)} \enspace, \label{eq:app1}
\end{eqnarray}
where the last inequality follows from Lemma \ref{lem:constantfactor}.
 
Let $M_{1}$ and $M_{2}$ be the $(1 - \epsilon)$-approximate mechanisms
from Lemma \ref{lem:fewitem} and Lemma \ref{lem:concentrated} respectively.
Consider the following mechanism $M$ for $[n]$:
\begin{enumerate}
\item Let the bidders submit their bids $\bb_{1,[n]}, \dots, \bb_{k,[n]}$.
\item Run $M_{1}$ on bids $\bb_{1,S_{1}},\bb_{2,S_{1}},\ldots,\bb_{k,S_{1}}$ and $M_{2}$ on bids $\bb_{1,S_{2}},\bb_{2,S_{2}},\ldots,\bb_{k,S_{2}}$.
\item Let $\bS'$, $\bp'$ and $\bS''$, $\bp''$ be the corresponding allocation and prices for $M_{1}$ and $M_{2}$, give items in $\bS'_{i}\cup\bS''_{i}$ to bidder $i$, and charge him $p'_{i}+p''_{i}$.
\end{enumerate}

Let $\calR(N)$ be the revenue for mechanism $N$, then $\calR(M)=\calR(M_{1})+\calR(M_{2})$. By Lemma \ref{lem:2}, we know 
$$\opt([n])\leq \opt(S_{1}) + \max_{j \in S_{2}} \E[\max_{i} v_{ij}] +
\max_{j \in S_{3}} \E[\max_{i} v_{ij}] \enspace.$$

First of all, the contribution of the last term is small according to
\eqref{eq:app1}. So it suffices to obtain revenue close to 
$\opt(S_{1})+\max_{j\in S_{2}} \E\left[\max_{i} v_{ij}\right]$.
Further, we know that $\calR(M_{1})\geq (1-\epsilon)\opt(S_{1})$ and
$\calR(M_{2})\geq (1-\epsilon)\sum_{j \notin S_{2}} \E[\max_i v_{ij}]$.
Therefore, we have $\calR(M)\geq (1-O(\epsilon))\opt([n])$.
    
Finally, the truthfulness (with respect to the corresponding solution
concept) of $M$ follows straightforwardly from the truthfulness of $M_1$
and $M_2$. So we have proved the theorem.
\end{proof}

\section{Nearly Optimal Mechanism for Constant Number of Items and
Bidders}\label{app:constlp}

\subsection{Discrete Explicit Access Model}

In this setting, the problem of optimal mechanism design for revenue
among IC and IR mechanisms or among BIC and IR mechanisms can be written
as polynomial-size linear programs (each bidder might have many different values for an item). Therefore, we can efficiently find
the optimal mechanism in these two settings. Since the LPs we used are very
standard (e.g.~see \cite{BBHM}), we will defer the discussion of these
LPs to Appendix \ref{app:lp}.

For the problem of optimal mechanism design among deterministically
truthful mechanisms, however, we need to solve the integer program
version of the LP of optimal IC mechanisms. In order to do so, we need
to reduce the size of the integer program from polynomial to constant.
We will take the standard approach of rounding down each bidder's value to
the nearest multiple of $\epsilon$.
As a result, for each bidder-item pair we only need to consider a
constant number of possible valuations. Recall there are only a constant
number of items and bidders, we can solve the constant-size integer
program for this coarsened support set efficiently. As a result of the
coarsening, however, we only get $\epsilon$-deterministically truthful
instead of perfect truthfulness.

\subsection{Continuous Oracle Access Model}

\subsubsection{DT Mechanism for a Single Bidder}

By the taxation principle, any deterministically truthful mechanism can
be interpreted as a bundle-pricing mechanism: the bidder is given a menu
of bundles of items such that the prices of the bundles are independent
on the reported values; moreover, the bidder always gets one of the
utility-maximizing bundles. In other words, it suffices to find the
nearly optimal bundle prices. In order to do so, we first show that in
order to obtain nearly optimal revenue it suffice to consider a finite
number of prices for each bundle via a standard price discretization lemma
attributed to Nisan (e.g.~see \cite{CHK}). Then, we can search over all
possible bundle-pricings within the discretized price set and choose the
optimal one. Since there are only constant number of items and thus
constant number of bundles, such brute-force search can be done
efficiently. For completeness we include a formal statement and the
proof of the price discretization lemma in Appendix \ref{app:pricedisc}.

\subsubsection{IC Mechanism for a Single Bidder} 

In this case, our starting
point is again the taxation principle. Any IC mechanism can be
interpreted as a lottery-pricing mechanism: the bidder is given a (not
necessarily finite) menu of lotteries, each of which is represented by a
vector of the probabilities of getting each item, such that the prices
of the lotteries are independent on the reported values; moreover, the
bidder always gets the utility maximizing lottery. By the same price
discretization lemma, we only need to consider a finite number of prices
for each lottery. However, there is an infinite number of possible
lotteries. We settle this problem by showing the lottery space can be
discretized as well. Concretely, we prove that in order to obtain $1 -
O(\epsilon)$ of the optimal revenue, it suffices to consider lotteries
in which the probabilities of getting each item are powers of $(1 +
\epsilon^2)$ and are greater than $\epsilon^2$. As a result, we can
combine the lottery discretization lemma
and the pricing discretization lemma to show that in order to get nearly
optimal lottery pricing it suffices to search over constant number of
lottery-pricing mechanisms and choose the best one. The proof of the
lottery discretization lemma is deferred to Appendix \ref{app:lotterydisc}.

\subsubsection{DT/IC/BIC Mechanisms for Multiple
Bidders} 

In order to solve the problem for multiple bidders, we use a
reduction to the discrete case: discretize the prior distributions by
rounding each sampled value to the closest powers of $(1 + \epsilon)$ and 
truncate values that are too large or too small according to the extreme
value theorem in \cite{CD}. We then find the nearly optimal mechanism
for the coarsened problem via the integer programing/linear programming
approach for the discrete case. Finally, we will round the bids of the
bidders to the closest powers of $(1 + \epsilon)$, run the above
mechanism on the coarsened bids, and use the allocation and prices
chosen by the mechanism. As a result of the rounding, the mechanisms we
obtain are only $\epsilon$-truthful with respect to the corresponding
solution concepts. Nonetheless, in the BIC case, we can use the
technique recently developed by Daskalakis and Weinberg \cite{DW} to
convert our $\epsilon$-BIC mechanism into a BIC one with only a small
additional loss in the expected revenue.
 
\section{Linear Programs for Multi-Item Auctions} \label{app:lp}

It has long been known that if the support set is finite, then the
problem of designing truthful (IC/BIC) mechanism that achieves optimal
revenue can be characterized by a linear problem. For completeness, we
will describe the standard linear programs for the multi-item auction. 

For any type profile $\bv$, any bidder
$i$, and any item $j$, we let $x(\bv)_{ij}$ denote the probability that
bidder $i$ gets item $j$ when the valuations are $\bv$, and let
$p(\bv)_i$ denote the expected payment of bidder $i$. The problem of
optimal multi-item auction among IC mechanisms has the following exact
LP characterization:
\begin{align*}
     \textrm{Maximize} & & \sum_{\bv} ~ \Pr[\bv] & \sum_{i=1}^k p(\bv)_i \quad \textrm{s.t.} \\
    \forall j, \bv : \quad & & \sum_{i=1}^k x(\bv)_{ij} & \le 1 \\
    \forall i, \bv, v'_i: \quad & & \sum_{j=1}^n x(\bv)_{ij} v_{ij} - p(\bv)_i & \ge \sum_{j=1}^n x(v'_i, v_{-i})_{ij} v_{ij} - p(v'_i, v_{-i})_i  \\
    \forall i, \bv : \quad & & \sum_{j=1}^n x(\bv)_{ij} v_{ij} - p(\bv)_i & \ge 0
\end{align*}

The LP characterization of the problem of optimal multi-item auction
among BIC mechanisms is almost the same, except for replacing the IC
constraints with the following BIC constraints for all $i$, $\bv$, and
$v_i'$:
$$\sum_{v_{-i}} \Pr[v_{-i}] \left( \sum_{j=1}^n x(v_i, v_{-i})_{ij} v_{ij} - p(v_i, v_{-i})_i \right) \ge \sum_{v_{-i}} \Pr[v_{-i}] \left( \sum_{j=1}^n x(v'_i, v_{-i})_{ij} v_{ij} - p(v'_i, v_{-i})_i \right) \enspace.$$

\section{Price Discretization Lemma} \label{app:pricedisc}

The following price discretization lemma is attributed to
Nisan (e.g., \cite{BBHM, CHK}):

\begin{lemma} \label{lem:pricedisc}
    For $\epsilon \in (0, 1)$, let $\bp$ and $\bp'$ be two bundle
    pricing schemes such that for any bundle $i$, $p_i \in [1-\epsilon,
    1 - \epsilon + \epsilon^2] p'_i$. Suppose the bidder buys bundle $j$
    when $\bp$ are the prices and buys bundle $\ell$ when $\bp'$ are the
    prices, then $p_j \ge (1 - 2\epsilon) p'_\ell$.
\end{lemma}

\begin{proof}
    By our assumption, we have $v_j - p_j \ge v_\ell - p_\ell$ and
    $v_\ell - p'_\ell \ge v_j - p'_j$. Summing up the two inequalities
    and cancelling the common terms, we have $p'_j - p_j \ge p'_\ell -
    p_\ell$. Note that by our assumption $p'_j - p_j \le p'_j - (1 -
    \epsilon) p'_j = \epsilon p'_j$, and $p'_\ell - p_\ell \ge p'_\ell -
    (1 - \epsilon + \epsilon^2) p'_\ell = (\epsilon - \epsilon^2)
    p'_\ell$. So we have $p'_j \ge (1 - \epsilon) p'_\ell$. Finally,
    $p_j \ge (1 - \epsilon) p'_j$. So $p_j \ge (1 - \epsilon)^2 p'_\ell
    \ge (1 - 2 \epsilon) p'_\ell$.
\end{proof}

By Lemma \ref{lem:pricedisc} we know that it suffices to consider prices that
are powers of $(1+\epsilon^2)$ in order to get $(1 - 2\epsilon)$ of the
optimal revenue. Of course, we still have infinite number of prices to
consider. In order to settle this problem, we will use the extreme value
theorem in \cite{CD} to conclude that for each bundle it suffices to
consider prices that are in a range whose upper and lower bounds differ
by at most an $O(\frac{1}{\epsilon} \log (\frac{1}{\epsilon}))$ factor
(this range may be different for different bundles). Therefore, we only
need to consider
$O\left( \log_{1+\epsilon} \left( \frac{1}{\epsilon}
\log(\frac{1}{\epsilon}) \right) \right) = O(\frac{1}{\epsilon}
\log(\frac{1}{\epsilon}))$ number of prices per bundle.  

\section{Lottery Discretization Lemma} \label{app:lotterydisc}

The following lottery discretization lemma is inspired by the idea in
the price discretization lemma. First, let us define some notations. We
will use a $n$-dimensional vector $\bq = (q_1, \dots, q_n)$ to denote a
lottery where $q_j$ is the probability of getting item $j$. A lottery
menu is a collection (may or may not be finite) of lottery-price pairs: 
$\{(\bq_1, p_1), (\bq_2, p_2), \dots \}$.

\begin{lemma} \label{lem:lotterydisc}
    Suppose $\epsilon \in (0, 1)$. Consider the optimal
    lottery menu $L$ and another lottery menu $L'$ obtained by
    rounding up probabilities of each lottery $(\bq_i, p_i) \in L$
    into $(\bq'_i, p_i)$ such that for all $j \in [n]$, $q'_{ij} \in
    [1+\epsilon-\epsilon^2, 1+\epsilon] q_{ij}$. Then, the expected
    revenue from menu $L'$ is at least a $(1 - O(\epsilon))$ fraction of
    that from menu $L$.
\end{lemma} 

\begin{proof}
    Suppose the type profile is $\bv$. Further, let us assume the
    bidder buys lottery $(\bq_j, p_j)$ when $L$ is presented and
    $(\bq'_\ell, p_\ell)$ when $L'$ is presented. We have 
    \begin{eqnarray} 
        \label{eq:1}
        \bv \cdot \bq_j - p_j & \ge & \bv \cdot \bq_\ell - p_\ell \\    
        \bv \cdot \bq'_\ell - p_\ell & \ge & \bv \cdot \bq'_j - p_j
        \notag
    \end{eqnarray}
    
    By summing up these
    two inequalities and cancelling the common terms, we have 
    $$\bv \cdot \bq'_\ell - \bv \cdot \bq_\ell \ge \bv \cdot \bq'_j - \bv \cdot \bq_j \enspace.$$
    
    By our assumption, we further have 
    $$\bv \cdot \bq'_j - \bv \cdot \bq_j \ge \bv \cdot (1 + \epsilon - \epsilon^2) \bq_j - \bv \cdot \bq_j = (\epsilon - \epsilon^2) \bv \cdot \bq_j \enspace,$$ 
    and
    $$\bv \cdot \bq'_\ell - \bv \cdot \bq_\ell \le \bv \cdot (1 + \epsilon) \bq_\ell - \bv \cdot \bq_\ell = \epsilon \bv \cdot \bq_\ell \enspace.$$
    
    Therefore, we have $\bv \cdot \bq_\ell \ge (1 - \epsilon) \bv \cdot
    \bq_j$. By this inequality and \eqref{eq:1}, we have 
    $$p_\ell \ge p_j + \bv \cdot \bq_\ell - \bv \cdot \bq_j \ge p_j - \epsilon \bv \cdot \bq_j \enspace.$$

    Hence, if we compare the expected revenue from $L'$, $\E[p_\ell]$,
    and the expected revenue of $L$, $\E[p_j]$, then the former is worse
    than the latter by no more than an $\epsilon$ fraction of the social
    welfare by $L$. We further note that the optimal social welfare and
    the optimal revenue differ by at most a constant factor. Thus, we
    have proved the lemma.
\end{proof}

By Lemma \ref{lem:lotterydisc}, we can round up the probabilities in each
lottery to some powers of $(1+\epsilon^2)$ in order to get $1 -
O(\epsilon)$ of the optimal revenue. There is only one catch in this
argument: by rounding up the probabilities, some of them may exceed $1$
and therefore become infeasible. We resolve this problem by rounding
down the probabilities as well as the prices of the resulting
discretized lotteries by a factor of $1 - \epsilon$. By doing so, we
retain feasibility with the extra cost of a $1 - \epsilon$ factor, but
we still gets $1 - O(\epsilon)$ of the optimal revenue.

Lemma \ref{lem:lotterydisc} reduces the number of lotteries from uncountably
infinite to countably infinite. We observe that we can further reduce
this number to finite by dropping invaluable lotteries and the
negligible entries in the valuable lotteries. Concretely, if the
expected value of a lottery
is at most an $\epsilon$ fraction of the expected welfare, then we can
ignore this lottery because the total revenue from such lotteries is at
most an $O(\epsilon)$ fraction of the optimal. Next, suppose we have a
lottery whose expected value is at most an $\epsilon$ fraction of the
expected social welfare. Then, any entry smaller than $\epsilon^2$
contributes at most an $O(\epsilon)$ fraction to the expected value of
this lottery, and hence can be dropped.

\section{Impossibility of Truthful Reserve-Revenue Mechanism}
\label{app:reserverevenue}

In this section, we will show that the reserve-revenue mechanisms can not
be implemented in a truthful and individually rational manner. First of
all, let us formally define the family of reserve-revenue mechanisms.

\begin{definition}
    A mechanism $M$ is a {\em reserve-revenue mechanism} if there is a
    reserve revenue $r^* > 0$ and a threshold social welfare $s^* > 0$
    (presumably $s^* \ge r^*$), such that $M$ achieves revenue at least
    $r^*$ whenever the social welfare (according to the bids) is at
    least $s^*$.
\end{definition}

We have the following negative result.

\begin{theorem} \label{thm:noreserverevenue}
    If there are at least $k \ge 2$ bidders and $n \ge 2$ items, then
    there are no truthful and individually rational reserve-revenue
    mechanism.
\end{theorem}

\begin{proof}
    It suffices to prove the theorem for the case of $k = n = 2$. Assume
    for contradiction that $M$ is a truthful and individually rational
    reserve-revenue mechanism with reserve revenue $r^*$ and threshold
    social welfare $s^*$.

    Let us consider what happens when $v_{11} = s^* - \frac{r^*}{3}$, 
    $v_{12} = 0$, $v_{21} = 0$, $v_{22} = s^* - \frac{r^*}{3}$. We claim
    that in this case $p_1, p_2 \le \frac{r^*}{3}$.

    Consider the alternative type profile in which bidder $1$'s values
    are $v'_{11} = \frac{r^*}{3}$ and $v'_{12} = 0$, and bidder $2$'s
    values are still the same. Note that the social
    welfare for this type profile is exactly $s^*$. So $M$ shall
    achieve revenue at least $r^*$. It is obvious that the only
    allocation that could achieve this level of revenue in an
    individually rational fashion is to give bidder $1$ item $1$ and to
    give bidder $2$ item $2$. The price for bidder $1$ in this case is at
    most $\frac{r^*}{3}$. By the taxation principle, from bidder $1$'s
    viewpoint any truthful mechanism should look like a menu of lotteries
    over possible outcomes with prices that do not depend on the value
    of bidder $1$. Moreover, bidder $1$ should always get one of the
    utility maximizing lottery. Therefore, we know that the lottery that
    corresponds to getting item $1$ and not getting item $2$ with
    probability $1$ is available to bidder $1$ with price at most
    $\frac{r^*}{3}$ when bidder $2$ bids $\bv_2$ and it is bidder $1$'s
    utility-maximizing lottery when her valuation is $\bv_1'$. Note that
    the only difference between $\bv_1$ and $\bv_1'$ is the value for
    item $1$ increases. So we conclude that when the type profile is
    $\bv_1$, $\bv_2$, bidder $1$ should purchase the same lottery with
    the same price. Hence, we have proved that $p_1 \le \frac{r^*}{3}$.

    Similarly, we can show that $p_2 \le \frac{r^*}{3}$. Now we get that
    the revenue achieved by $M$ when the values are $\bv_1$ and $\bv_2$
    is at most $p_1 + p_2 \le \frac{2 r^*}{3} < r^*$. Thus, we have
    obtained a contradiction.
\end{proof}

\begin{remark}
    The conditions in Theorem \ref{thm:noreserverevenue} cannot be relaxed for
    that if $k = 1$, then the grand-bundle-reserve-price auction is a
    reserve-revenue mechanism, and if $n = 1$, then the standard reserve
    price auction is a reserve-revenue mechanism.
\end{remark}

\end{document}